\documentclass[fleqn]{llncs}

\usepackage[firstpage]{draftwatermark}\SetWatermarkText{\hspace*{5.2in}\raisebox{5in}{\includegraphics[scale=0.1]{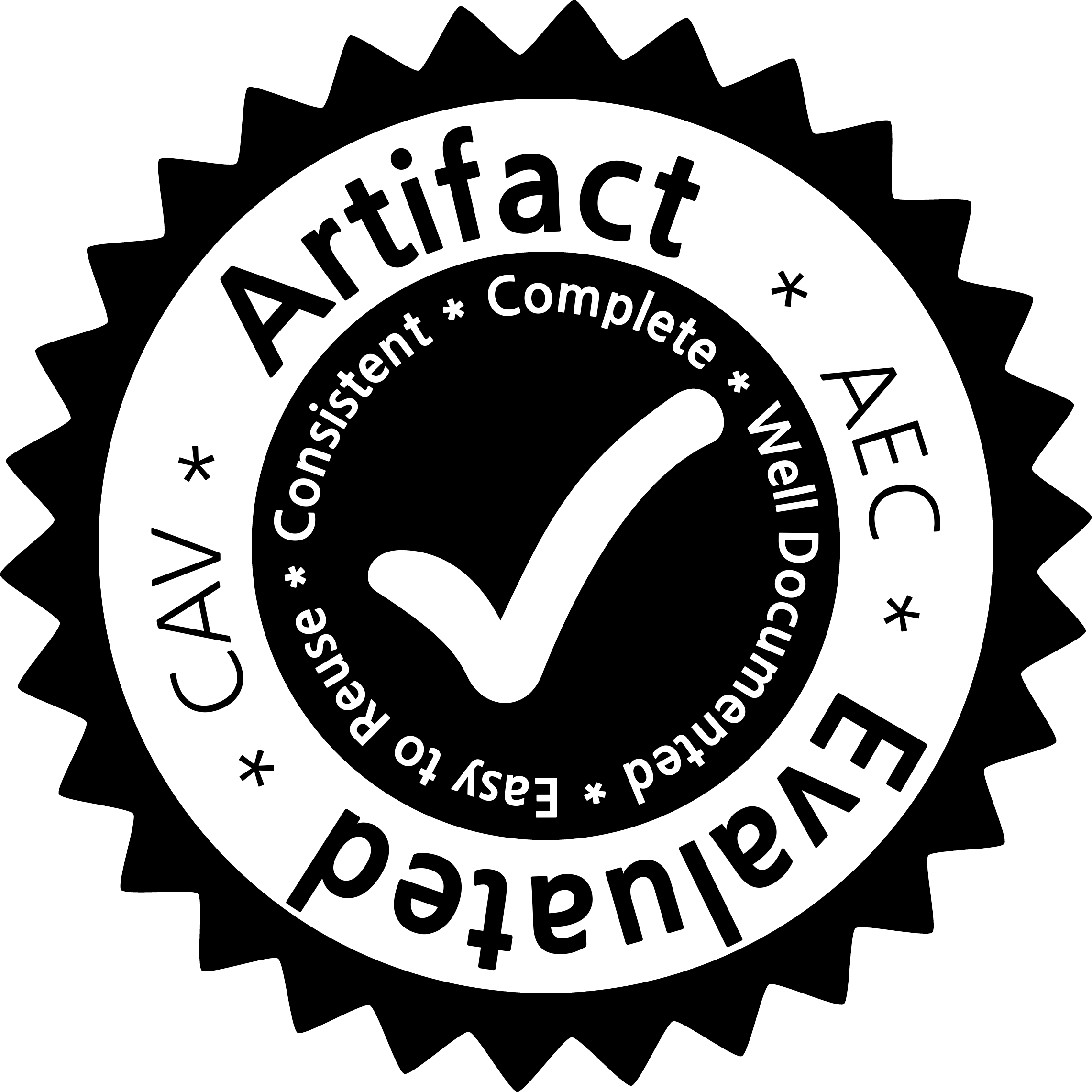}}}\SetWatermarkAngle{0}

\usepackage{verbatim}

\usepackage{amsmath}
\usepackage{amssymb}
\usepackage{stmaryrd} 
\usepackage{wasysym}
\usepackage{mathtools} 
\usepackage{bm}

\usepackage{proof}

\usepackage[noend]{algpseudocode}
\algrenewcommand\algorithmicindent{1.3em}
\usepackage{algorithm}

\usepackage{multirow}

\usepackage{cite}

\usepackage{subfig}

\usepackage{tikz}
\usetikzlibrary{arrows,automata,positioning}
\usetikzlibrary{decorations.pathmorphing}
\usepgflibrary{shapes.geometric}
\usepackage{pgfplots}
\pgfplotsset{compat=1.14}

\newcommand{\set}[1]{{\{#1\}}}
\newcommand{\ldot}{\mathpunct{.}}
\newcommand{\dom}{\mathrm{dom}}

\newcommand{\assignment}{\alpha}
\newcommand{\quant}{Q\,}
\newcommand{\bool}{\mathbb{B}}

\newcommand{\literals}{\mathit{lit}}

\newcommand{\pobj}{\mathcal{P}}
\newcommand{\resrule}{\mathrm{res}}
\newcommand{\initrule}{\mathrm{init}}
\newcommand{\redrule}{\forall\mathrm{red}}
\newcommand{\expresrule}{\forall\mathrm{exp\text{-}res}}

\newcommand{\tree}{\mathcal{T}}

\newcommand{\expandvar}{\mathit{expand\text{-}var}}
\newcommand{\expand}{\mathit{expand}}

\newcommand{\xor}{\mathrm{xor}}
\newcommand{\bfalse}{\bot}
\newcommand{\btrue}{\top}

\newcommand{\qres}{Q\text{-resolution}}
\newcommand{\expres}{\forall\text{Exp+Res}}
\newcommand{\redres}{\forall\text{Red+Res}}
\newcommand{\redexpres}{\forall\text{Red+}\forall\text{Exp+Res}}
\newcommand{\ircalc}{\text{IR-Calc}}

\newcommand{\crn}{\mathrm{CR}_n}
\newcommand{\qparity}{\mathrm{QParity}}
\newcommand{\dagn}{\mathrm{DAG}_n}

\newcommand{\caqe}{\text{CAQE}}
\newcommand{\rareqs}{\text{RAReQS}}
\newcommand{\qesto}{\text{Qesto}}
\newcommand{\depqbf}{\text{DepQBF}}
\newcommand{\ghostq}{\text{GhostQ}}
\newcommand{\bloqqer}{\text{Bloqqer}}

\newcommand{\minisat}{\text{MiniSat}}
\newcommand{\picosat}{\text{PicoSAT}}
\newcommand{\cmsat}{\text{cryptominisat}}
\newcommand{\lingeling}{\text{Lingeling}}

\begin{document}
 
\title{On Expansion and Resolution in CEGAR Based QBF Solving\thanks{Supported by the European Research Council (ERC) Grant OSARES (No.~683300).}}

\author{Leander Tentrup}

\institute{Saarland University, Saarbr\"ucken, Germany\\\email{tentrup@react.uni-saarland.de}}

\maketitle

\begin{abstract}
A quantified Boolean formula~(QBF) is a propositional formula extended with universal and existential quantification over propositions.
There are two methodologies in CEGAR based QBF solving techniques, one that is based on a refinement loop that builds partial expansions and a more recent one that is based on the communication of satisfied clauses.
Despite their algorithmic similarity, their performance characteristics in experimental evaluations are very different and in many cases orthogonal.
We compare those CEGAR approaches using proof theory developed around QBF solving and present a unified calculus that combines the strength of both approaches.
Lastly, we implement the new calculus and confirm experimentally that the theoretical improvements lead to improved performance.
\end{abstract}

\section{Introduction}

Efficient solving techniques for Boolean theories are an integral part of modern verification and synthesis methods.
Especially in synthesis, the amount of choice in the solution space leads to propositional problems of enormous size.
Quantified Boolean formulas~(QBFs) have repeatedly been considered as a candidate theory for synthesis approaches~\cite{journals/corr/FinkbeinerT15,conf/mtv/MillerSB13,conf/vmcai/BloemKS14,conf/fmcad/BloemEKKL14,conf/birthday/Finkbeiner15,conf/tacas/FaymonvilleFRT17} and recent advances in QBF solvers give rise to hope that QBF may help to increase the scalability of those approaches.

Solving quantified Boolean formulas~(QBF) using partial expansions in a counterexample guided abstraction and refinement (CEGAR) loop~\cite{journals/ai/JanotaKMC16} has proven to be very successful.
From its introduction, the corresponding solver $\rareqs$ won several QBF competitions.
In recent work, a different kind of CEGAR algorithms have been proposed~\cite{conf/ijcai/JanotaM15,conf/fmcad/RabeT15}, implemented in the solvers $\qesto$ and $\caqe$.
All those CEGAR approaches share algorithmic similarities like working recursively over the structure of the quantifier prefix and using SAT solver to enumerate candidate solutions.
However, instead of using partial expansions of the QBF as $\rareqs$ does, newer approaches base their refinements on whether a set of clauses is satisfied or not.
Despite those algorithmic similarities, the performance characteristics of the resulting solver in experimental evaluations are very different and in many cases orthogonal: While $\rareqs$ tends to perform best on instances with a low number of quantifier alternations, $\qesto$ and $\caqe$ have an advantage in instances with many alternations~\cite{conf/fmcad/RabeT15}.

Proof theory has been repeatedly used to improve the understanding of different solving techniques.
For example, the proof calculus $\expres$~\cite{journals/tcs/JanotaM15} has been developed to characterize aspects of expansion-based solving.
In this paper, we introduce a new calculus $\redres$ that corresponds to the clausal-based CEGAR approaches~\cite{conf/ijcai/JanotaM15,conf/fmcad/RabeT15}.
The levelized nature of those algorithms are reflected by the rules of this calculus, universal reduction and propositional resolution, which are applied to blocks of quantifiers.
We show that this calculus is inherently different to $\expres$ explaining the empirical performance results.
In detail, we show that $\redres$ polynomial simulates level-ordered $\qres$.
We also discuss an extension to $\redres$ that was already proposed as solving optimizations~\cite{conf/fmcad/RabeT15} and show that this extension makes the resulting calculus exponential more concise.

Further, we integrate the $\expres$ calculus as a rule that can be used within the $\redres$ calculus, leading to a unified proof calculus for all current CEGAR approaches.
We show that the unified calculus is exponential stronger than both $\expres$ and $\redres$, as well as just applying both simultaneously.
This unified calculus serves as a base for implementing an expansion refinement in the QBF solver $\caqe$.
On standard benchmark sets, the combined approach leads to a significant empirical improvement over the previous implementation.

\section{Preliminaries}

\subsection{Quantified Boolean Formulas}

We consider quantified Boolean formulas in prenex conjunctive normal form (PCNF), that is a formula consisting of a linear and consecutive quantifier prefix as well as a propositional matrix.
A \emph{matrix} is a set of clauses, and a clause is a disjunctive combination of \emph{literals} $l$, that is either a variable or its negation.

Given a clause $C = (l_1 \lor l_2 \lor \ldots \lor l_n)$, we use set notation interchangeably, that is $C$ is also represented by the set $\set{l_1,l_2,\dots,l_n}$.
Furthermore, we use standard set operations, such as union and intersection, to work with clauses.

For readability, we lift the quantification over variables to the quantification over sets of variables and denote a maximal consecutive block of quantifiers of the same type $\forall x_1\ldot \forall x_2\ldot \cdots\forall x_n\ldot\varphi$ by $\forall X\ldot\varphi$ and $\exists x_1\ldot\exists x_2\ldot\cdots\exists x_n\ldot\varphi$ by $\exists X\ldot\varphi$, accordingly, where $X=\set{x_1,\dots,x_n}$. 

Given a set of variables $X$, an \emph{assignment} of $X$ is a function $\alpha : X \rightarrow \bool$ that maps each variable $x \in X$ to either true ($\btrue$) or false ($\bfalse$).
When the domain of~$\alpha$ is not clear from context, we write $\alpha_X$.
We use the instantiation of a QBF~$\Phi$ by assignment $\alpha$, written $\Phi[\alpha]$ which removes quantification over variables in $\dom(\alpha)$ and replaces occurrences of $x \in \dom(\alpha)$ by $\alpha(x)$.
We write $\alpha \vDash \varphi$ if the assignment $\alpha$ satisfies a propositional formula $\varphi$, i.e., $\varphi[\alpha] \equiv \btrue$.

\subsection{Resolution}

Propositional resolution is a well-known method for refuting propositional formulas in conjunctive normal form~(CNF).
The resolution rule allows to \emph{merge} two clauses that contain the same variable, but in opposite signs.
\begin{align*}
  &\infer[\mathrm{res}]{C \cup C'}{C \cup \set{l} & C' \cup \set{\overline{l}}}
\end{align*}
A resolution proof $\pi$ is a series of applications of the resolution rule.
A propositional formula is unsatisfiable if there is a resolution proof that derives the empty clause.
We visualize resolution proofs by a graph where the nodes with indegree 0 are called the leaves and the unique node with outdegree 0 is called the root.
We depict the graph representation of a resolution proof in Fig.~\ref{fig:resolution}\subref{fig:resolution-example}.
The \emph{size} of a resolution proof is the number of nodes in the graph.

\begin{figure}[t]
  \centering
  \subfloat[Resolution rule]{
    \begin{tikzpicture}[thick]
      \tikzstyle{node}=[draw, rectangle, rounded corners, inner sep=5pt]
      \node[node] (root) {$C \cup C'$};
      \node[node,above left=of root] (lhs) {$C \cup \set{l}$};
      \node[node,above right=of root] (rhs) {$C' \cup \set{\overline{l}}$};
      
      \draw (root) -- (lhs)
            (root) -- (rhs)
            ;
    \end{tikzpicture}
    \label{fig:resolution-rule}
  }
  \subfloat[Resolution proof for formula $(c_1)(c_2)(\overline{c}_1 \lor \overline{c}_2)$]{
    \begin{tikzpicture}[thick]
      \tikzstyle{node}=[draw, rectangle, rounded corners, inner sep=5pt]
      \node[node] (root) {$\bfalse$};
      \node[node,above left=0.5 and 0 of root] (lhs) {$\overline{c}_1$};
      \node[node,above right=0.5 and 0 of root] (rhs) {$c_1$};
      \node[node,above left=0.5 and 0 of lhs] (lhslhs) {$\overline{c}_1 \lor \overline{c}_2$};
      \node[node,above right=0.5 and 0 of lhs] (lhsrhs) {$c_2$};
      
      \draw (root) -- (lhs)
            (root) -- (rhs)
            (lhs) -- (lhslhs)
            (lhs) -- (lhsrhs)
            ;
      \node[left=2 of root] {};
      \node[right=2 of root] {};
    \end{tikzpicture}
    \label{fig:resolution-example}
  }
  \caption{Visualization of the resolution rule as a graph.}
  \label{fig:resolution}
\end{figure}
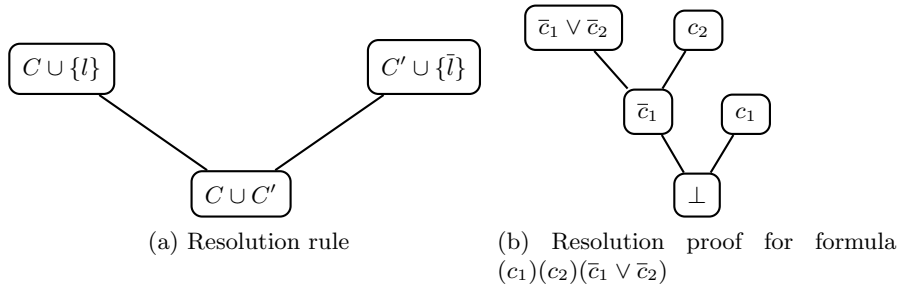

\subsection{Proof Systems}

We consider proof systems that are able to refute quantified Boolean formulas.
To enable comparison between proof systems, one uses the concept of \emph{polynomial simulation}.
A proof system $P$ polynomially simulates ($p$-simulates) $P'$ if there is a polynomial $p$ such that for every number $n$ and every formula $\Phi$ it holds that if there is a proof of $\Phi$ in $P'$ of size $n$, then there is a proof of $\Phi$ in $P$ whose size is less than $p(n)$.
We call $P$ and $P'$ polynomial equivalent, if $P'$ additionally $p$-simulates $P$.

A refutation based calculus (such as resolution) is regarded as a proof system because it can refute the negation of a formula.

Figure~\ref{fig:overview-proof-systems} gives an overview over the proof systems introduced in this paper and their relation.
An edge $P \rightarrow P'$ means that $P$ $p$-simulates $P'$ (transitive edges are omitted).
A dashed line indicates incomparability results.

\begin{figure}[t]
  \centering
  \begin{tikzpicture}[->,>=stealth',shorten >=1pt,auto,thick,scale=1,transform shape]
  \tikzstyle{state}=[draw,rectangle,rounded corners,inner sep=5pt,align=center]
  \node[state,fill=black!5] (redexpres-strengthen) {$\redexpres$\\strengthen};
  \node[state,fill=black!5,below left=1.1 and -1 of redexpres-strengthen] (redexpres) {$\redexpres$};
  \node[state,below=1.2 of redexpres] (expres) {$\expres$};
  \node[state,fill=black!5,below right=1 and -0.5 of redexpres-strengthen] (redres-strengthen) {$\redres$\\strengthen};
  \node[state,fill=black!5,below=of redres-strengthen] (redres) {$\redres$};
  \node[state,right=1.2 of redres] (level-ordered) {level-ordered\\$\qres$};
  \node[state,above right=1 and -0.7 of level-ordered] (qres) {$\qres$};
  \node[state,below right=1 and -0.7 of qres] (tree-like) {tree-like\\$\qres$};
  
  \draw (redexpres-strengthen) edge (redexpres)
  		(redexpres) edge node[swap] {\tiny Theorem~\ref{thm:expres-vs-redexpres}} (expres)
  		(redexpres) edge node[swap] {\tiny Theorem~\ref{thm:combined-stronger-than-applying-independently}} (redres)
        (redexpres-strengthen) edge (redres-strengthen)
        (redres-strengthen) edge node {\tiny Theorem~\ref{thm:separation-strengthen-rule}} (redres)
        (redres) edge[<->] node {\tiny Theorem~\ref{thm:redres-level-ordered-qres-equivalent}} (level-ordered)
        (qres) edge (level-ordered)
        (qres) edge (tree-like)
        ;
  
  \coordinate (a) at (0,-5.5);
  \coordinate (b) at (5,-5.5);
  \draw[->] (expres) .. controls (a) and (b) .. node {\tiny\cite{journals/tcs/JanotaM15}} (tree-like);
  
  \draw[-,dashed] (redexpres-strengthen) -- node {\tiny Theorem~\ref{thm:redexpres-strengthen-qres-incomparable}} (qres);
  \draw[-,dashed] (redexpres) -- node {\tiny Theorem~\ref{thm:extensions-incomparable}} (redres-strengthen);
  \draw[-,dashed] (level-ordered) -- node {\tiny\cite{journals/ipl/MahajanS16}} (tree-like);
  \draw[-,dashed] (expres) -- node {\tiny\cite{journals/ipl/MahajanS16}, Theorem~\ref{thm:redres-level-ordered-qres-equivalent}} (redres);
\end{tikzpicture}
  \caption{Overview of the proof systems and their relations. Solid arrows indicate $p$-simulation relation. Dashed lines indicate incomparability results. The gray boxes are the ones introduced in this paper.}
  \label{fig:overview-proof-systems}
\end{figure}
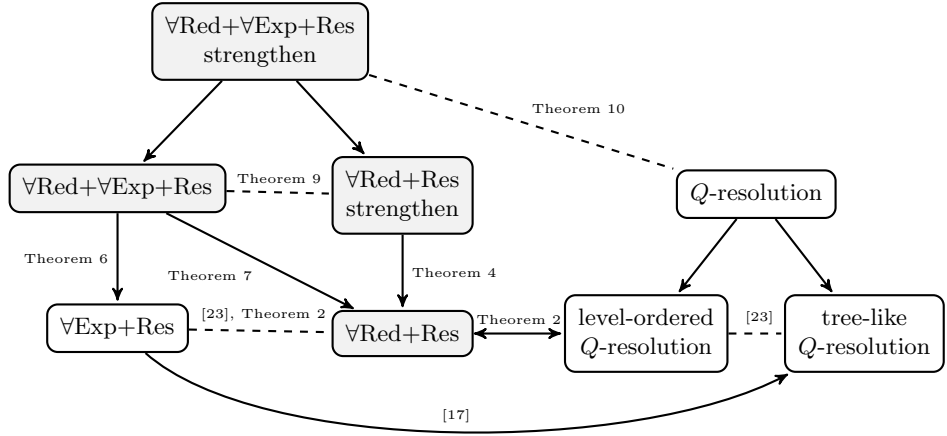

\section{Proof Calculi}

Given a PCNF formula $\quant X_1 \dots \quant X_n \ldot \bigwedge_{1 \leq i \leq m} C_i$.
We define a function $\literals(i, k)$ that returns the literals of clause $C_i$ that are bound at quantifier level $k$ ($1 \leq k \leq n$).
Further, we generalize this definition to $\literals(i, > k)$ and $\literals(i, < k)$ that return the literals bound after (before) level $k$.
We define $\literals(i, 0) = \literals(i, n+1) = \emptyset$ for every $1 \leq i \leq m$.
We use $\mathcal{C}$ to denote a set of clauses and $Q_k \in \set{\exists,\forall}$ to denote the quantification type of level $k$.

\subsection{A Proof System for Clausal Abstractions}

We start by defining the object on which our proof system $\redres$ is based on.
A \emph{proof object} $\pobj^k$ consists of a set of indices $\pobj$ where an index $i \in \pobj$ represents the $i$-th clause in the original matrix and $k$ denotes the $k$-th level of the quantifier hierarchy.
We define an operation $\literals(\pobj^k) = \bigcup_{i \in \pobj} \literals(i, k)$, that gives access to the literals of clauses contained in $\pobj^k$.
The leaves in our proof system are singleton sets $\set{i}^z$ where $z$ is the maximum quantification level of all literals in clause $C_i$.
The root of a refutation proof is the proof object $\pobj^0$ that represents the empty set, i.e., $\literals(\pobj^0) = \emptyset$.

The rules of the proof system is given in Fig.~\ref{fig:unsat_proof_system}.
It consists of three rules, an axiom rule ($\initrule$) that generates leaves, a resolution rule ($\resrule$), and a universal reduction rule ($\redrule$).
The latter two rules enable to transform a premise that is related to quantifier level $k$ into a conclusion that is related to quantifier level~$k-1$.
The universal reduction rule and the resolution rule are used for universal and existential quantifier blocks, respectively.

\begin{figure}[t]
  \begin{equation*}
    \begin{array}{ll}
      \infer[\resrule]
      { \left( \bigcup_{i \in \set{1,\dots,j}} \pobj_i \right)^{k-1} }
      {\pobj_{1}^k & \cdots & \pobj_{j}^k & \pi}
      \qquad
      &
      \begin{minipage}[c]{0.5\textwidth}
        $Q_k = \exists$\\
        $\pi \text{ is a resolution refutation proof for } \bigwedge_{1 \leq i \leq j} \literals(\pobj_i^k)$
      \end{minipage}
      \\ \\
      \infer[\redrule]
      { \pobj^{k-1} }
      { \pobj^k }
      &
      \begin{minipage}[c]{0.5\textwidth}
        $Q_k = \forall$ \\
        $\forall l \in \literals(\pobj^k) \ldot \overline{l} \notin \literals(\pobj^k)$
      \end{minipage}
      \\\\
      \infer[\initrule]{\set{i}^k}{}
      &
      \begin{minipage}[c]{0.5\textwidth}
        $1 \leq i \leq m$ \\
        $\literals(i, >k) = \emptyset$
      \end{minipage}
      \\
    \end{array}
  \end{equation*}
  \caption{The rules of the $\redres$ calculus.}
  \label{fig:unsat_proof_system}
\end{figure}

\paragraph{Resolution rule.}
There is a close connection between ($\resrule$) and the propositional resolution as ($\resrule$) merges a number of proof objects $\pobj_i^k$ of level $k$ into a single proof object of level $k-1$.
It does so by using a resolution proof for a propositional formula that is constructed from the premises $\pobj_i^k$.
This propositional formula $\bigwedge_{1 \leq i \leq j} \literals(\pobj_i^k)$ contains \emph{only} literals of level $k$.
Intuitively, this rule can be interpreted as follows: a resolution proof over those clauses rules out any possible existential assignment at quantifier level $k$, thus, one of those clauses has to be satisfied at an earlier level.

\paragraph{Universal reduction rule.}
In contrast to ($\resrule$), ($\redrule$) works on single proof objects.
It can be applied if level $k$ is universal and the premise does not encode a universal tautology, i.e., for every literal $l \in \literals(\pobj^k)$, the negated literal $\overline{l}$ is not contained in $\literals(\pobj^k)$.

\paragraph{Graph representation.}
A proof in the $\redres$ calculus can be represented as a directed acyclic graph~(DAG).
The nodes in the DAG are proof objects~$\pobj^k$ and the edges represent applications of ($\resrule$) and ($\redrule$).
The rule ($\resrule$) is represented by a hyper-edge that is labeled with the propositional resolution proof $\pi$.
Edges representing the universal reduction can thus remain unlabeled without introducing ambiguity.
The \emph{size} of a $\redres$ proof is the number of nodes in the graph together with the number of inner (non-leaf, non-root) nodes of the containing propositional resolution proofs.

A \emph{refutation} in the $\redres$ calculus is a proof that derives a proof object $\pobj^0$ at level $0$.
A proof for some $\pobj^k$ is a $\redres$ proof with root $\pobj^k$.
Thus, a proof for $\pobj^k$ can be also viewed as a refutation for the formula $\quant X_{k+1} \dots \quant X_n \ldot\allowbreak \bigwedge_{i \in \pobj} \literals(i,>k)$ starting with quantifier level $k+1$ and containing clauses represented by~$\pobj$.

\begin{example} \label{ex:example_redres_proof}
  Consider the following QBF
  \begin{equation} \label{eq:example_qbf}
    \underbrace{\exists e_1}_1 \ldot
    \underbrace{\forall u_1}_2 \ldot
    \underbrace{\exists c_1, c_2}_3 \ldot
    ( \underbrace{\overline{e}_1 \lor c_1}_{C_1} )
    ( \underbrace{\overline{u}_1 \lor c_1}_{C_2} )
    ( \underbrace{e_1 \lor c_2}_{C_3} )
    ( \underbrace{u_1 \lor c_2}_{C_4} )
    ( \underbrace{\overline{c}_1 \lor \overline{c}_2}_{C_5} ) \enspace.
  \end{equation}
  
  The refutation in the $\redres$ calculus is given in Fig.~\ref{fig:example_redres}.
  In the nodes, we represent the proof objects $\pobj^k$ in the first component and the represented clause in the second component.
  The proof follows the structure of the quantifier prefix, i.e., it needs four levels to derive a refutation.
  The resolution proof~$\pi_1$ for propositional formula
  \begin{equation*}
    \literals(\set{1}^3) \land
    \literals(\set{4}^3) \land
    \literals(\set{5}^3)
    {} \equiv (c_1)(c_2)(\overline{c}_1 \lor \overline{c}_2)  
  \end{equation*}
  is depicted in Fig.~\ref{fig:resolution}\subref{fig:resolution-example}.
  
  \begin{figure}[t]
    \centering
    \begin{tikzpicture}[>=stealth',shorten >=1pt,auto,thick,scale=0.7,transform shape]
  \tikzstyle{proofobject}=[draw,rectangle,rounded corners,inner sep=5pt,align=center,font=\Large]
  \tikzstyle{hyper}=[draw,rectangle,inner sep=5pt,align=center,font=\large]
  
  \node[proofobject] (root) {$(\set{1,2,3,4,5}^0,\bot)$};
  
  \node[above=of root,hyper] (hyper1) {$\pi_3$};
  \node[above left=of hyper1,proofobject] (c1_1) {$(\set{1,4,5}^1,\overline{e}_1)$};
  \node[above right=of hyper1,proofobject] (c3_1) {$(\set{2,3,5}^1,e_1)$};
  
  \node[above=of c1_1,proofobject] (c14_2) {$(\set{1,4,5}^2,u_1)$};
  \node[above=of c3_1,proofobject] (c23_2) {$(\set{2,3,5}^2,\overline{u}_1)$};
  
  \node[above=of c14_2,hyper] (hyper2) {$\pi_1$};
  \node[above=of c23_2,hyper] (hyper3) {$\pi_2$};
  
  \node[above left=of hyper2,proofobject] (c1_3) {$(\set{1}^3,c_1)$};
  \node[above=of hyper2,proofobject] (c4_3) {$(\set{4}^3,c_2)$};
  \node[above right=of hyper2,proofobject] (c5_3) {$(\set{5}^3,\overline{c}_1 \lor \overline{c}_2)$};
  
  \node[above right=of hyper3,proofobject] (c3_3) {$(\set{3}^3,c_2)$};
  \node[above=of hyper3,proofobject] (c2_3) {$(\set{2}^3,c_1)$};
  
  \draw (c1_3) edge (hyper2)
        (c4_3) edge (hyper2)
        (c5_3) edge (hyper2)
        
        (c5_3) edge (hyper3)
        (c2_3) edge (hyper3)
        (c3_3) edge (hyper3)
        
        (hyper2) edge (c14_2)
        (hyper3) edge (c23_2)
        
        (c14_2) edge (c1_1)
        (c23_2) edge (c3_1)
        
        (c1_1) edge (hyper1)
        (c3_1) edge (hyper1)
        
        (hyper1) edge (root)
        ;
\end{tikzpicture}
    \caption{A $\redres$ refutation for formula~\eqref{eq:example_qbf}.}
    \label{fig:example_redres}
  \end{figure}
\end{example}

In the following, we give a formal correctness argument and compare our calculus to established proof systems.
A QBF proof system is \emph{sound} if deriving a proof implies that the QBF is false and it is \emph{refutational complete} if every false QBF has a proof.

\begin{theorem} \label{thm:redres-soundness}
  $\redres$ is sound and refutational complete for QBF.
\end{theorem}
\begin{proof}
  The completeness proof is carried out by induction over the quantifier prefix.
  
  \emph{Induction base.}
  Let $\exists X \ldot \varphi$ be a false QBF and $\varphi$ be propositional.
  Then $(\resrule)$ derives some $\pobj^0$ because resolution is complete for propositional formulas.
  Let $\forall X \ldot \varphi$ be a false QBF and $\varphi$ be propositional.
  Picking an arbitrary (non-tautological) clause $C_i$ and applying $(\redrule)$ leads to $\set{i}^0$.
  
  \emph{Induction step.}
  Let $\exists X \ldot \Phi$ be a false QBF, i.e., for all assignments $\assignment_X$ the QBF $\Phi[\assignment_X]$ is false.
  Hence, by induction hypothesis, there exists a $\redres$ proof for every $\Phi[\assignment_X]$.
  We transform those proofs in a way that they can be used to build a proof for $\Phi$.
  Let $P$ be a proof of $\Phi[\assignment_X]$.
  $P$ has a distinct root node (representing the empty set), that was derived using $(\redrule)$ as $\Phi[\assignment_X]$ starts with a universal quantifier.
  To embed $P$ in $\Phi$, we increment every level in $P$ by one, as $\Phi$ has one additional (existential) quantifier level.
  Then, instead of deriving the empty set, the former root node derives a proof object of the form $\pobj^1$.
  Let $N$ be the set of those former root nodes.
  By construction, there exists a resolution proof $\pi$ such that the empty set can be derived by $(\resrule)$ using $N$ (or a subset thereof).
  Assuming otherwise leads to the contradiction that some $\Phi[\assignment_X]$ is true.
  
  Let $\forall X \ldot \Phi$ be a false QBF, i.e., there is an assignment $\assignment_X$ such that the QBF $\Phi[\assignment_X]$ is false.
  Hence, by induction hypothesis, there exists a $\redres$ proof for $\Phi[\assignment_X]$.
  Applying $(\redrule)$ using $\assignment_X$ is a $\redres$ proof for $\Phi$.
  
  For soundness it is enough to show that one cannot derive a clause using this calculus that changes the satisfiability.
  Let $\Phi = \quant X_1 \dots \quant X_n \ldot \bigwedge_{1 \leq i \leq m} C_i$ be an arbitrary QBF.
  For every level $k$ and every $\pobj^k$ generated by the application of the $\redres$ calculus, it holds that $\Phi$ and $\quant X_1 \dots \quant X_n \ldot \bigwedge_{1 \leq i \leq m} C_i \land (\bigvee_{i \in \pobj} \bigvee_{l \in \literals(i, \leq k)} l)$ are equisatisfiable.
  Assume otherwise, then either $(\redrule)$ or $(\resrule)$ have derived a $\pobj^k$ that would make $\Phi$ false.
  Again, by induction, one can show that if $(\redrule)$ derived a $\pobj^k$ that makes $\Phi$ false, the original premise $\pobj^{k+1}$ would have made $\Phi$ false; likewise, if $(\resrule)$ derived a $\pobj^k$ that makes $\Phi$ false, the conjunction of the premises have made $\Phi$ false.
  \qed
\end{proof}

\paragraph{Comparison to $\qres$ calculus.}
$\qres$~\cite{journals/iandc/BuningKF95} is an extension of the (propositional) resolution rule to handle universal quantification.
The universal reduction rule allows the removal of universal literal $u$ from a clause $C$ if no existential literal $l \in C$ depends on $u$.
There are also additional rules on when the resolution rule can be applied, i.e., it is not allowed to produce tautology clauses using the resolution rule.
The definitions of $Q$-resolution proof and refutation are analogous to the propositional case.

There are two restricted classes of $\qres$ that are commonly considered, that is \emph{level-ordered} and \emph{tree-like} $\qres$.
A $\qres$ proof is level-ordered if resolution of an existential literal $l$ at level $k$ happens before every other existential literal with level $< k$.
A $\qres$ proof is tree-like if the graph representing the proof has a tree shape.

As a first result, we show that $\redres$ is polynomially equivalent to level-ordered $\qres$, i.e., a proof in our calculus can be polynomially simulated in level-ordered $\qres$ and vice versa.
While this is straightforward from the definitions of both calculi, this is much less obvious if one looks at the underlying algorithms of the CEGAR approaches~\cite{conf/ijcai/JanotaM15,conf/fmcad/RabeT15} and QCDCL~\cite{conf/iccad/ZhangM02}.

\begin{theorem} \label{thm:redres-level-ordered-qres-equivalent}
  $\redres$ and level-ordered $\qres$ are $p$-sim. equivalent.
\end{theorem}
\begin{proof}
  A $\redres$ proof can be transformed into a $Q$-resolution proof by replacing every node $\pobj^k$ by the clause $(\bigvee_{i \in \pobj} \bigvee_{l \in \literals(i, \leq k)} l)$ and by replacing the hyper-edge labeled with $\pi$ by a graph representing the applications of the resolution rule.
  Similarly, a level-ordered $Q$-resolution proof can be transformed into a $\redres$ proof by a step-wise transformation from leaves to the root.
  This way, one can track the clauses needed for constructing the proof objects $\pobj^k$ at every level $k$.
  \qed
\end{proof}

Despite being equally powerful, the differences are important and enable the expansion based extension that we will introduce in the next section.
One difference is that our calculus only reasons about literals of one quantifier level, which allows us to use plain resolution without any changes (as are needed in $\qres$).
Further, the proof rules capture the fact that only proof obligations are communicated between the quantifier levels of the QBF.
An immediate consequence is that every refutation in the proof system is DAG-like and has exactly depth $k+1$.

Since the level-ordering constraint imposes an order on the resolution, the size of the refutation proof may be exponentially larger for some formulas~\cite{journals/amai/Goerdt92}.
Hence, also $\redres$ is in general exponentially weaker than unrestricted $\qres$.
In practice, and already noted by Janota and Marques-Silva~\cite{journals/tcs/JanotaM15}, solvers that are based on $\qres$ proofs produce level-ordered $\qres$.

In the initial version of $\caqe$~\cite{conf/fmcad/RabeT15} an optimization that can generate new resolvents at level $k$ without recursion into deeper levels was described.
We model this optimization as a new rule extending the $\redres$ calculus and show that this rule leads to an exponential separation.

\paragraph{Strong UNSAT Rule.}

In the implementation of $\caqe$, we used an optimization which we called \emph{strong UNSAT refinement}~\cite{conf/fmcad/RabeT15}, that allowed the solver to strengthen a certain type of refinements.
The basic idea behind this optimization is that if the solver determines that, at an existential level $k$, a certain set of clauses $\mathcal{C}$ cannot be satisfied at the same time, then every alternative set of clauses $\mathcal{C}'$, that is equivalent with respect to the literals in levels $>k$, cannot be satisfied as well.
We introduce the following proof rule that formalizes this intuition.
We extend proof objects $\pobj^k$ such that they can additionally contain fresh literals, i.e., literals that were not part of the original QBF.
Those literals are treated as they were bound at level $k$, i.e., they are contained in $\literals(\pobj^k)$ and can thus be used in the premise of the rule $(\resrule)$, but are not contained in the conclusion~$\pobj^{k-1}$.
\begin{center}
\begin{minipage}{0.7\textwidth}
  $\infer[\mathrm{strengthen}]
    {(\set{a} \cup \pobj)^k \quad \set{\overline{a}, j_1}^k \quad\cdots\quad \set{\overline{a}, j_n}^k}
    {(\pobj \cup \set{ i })^k}$
\end{minipage}%
\begin{minipage}{0.3\textwidth}
    $Q_k = \exists$,\\
    $\literals(j,>k) \subseteq \literals(i,>k)$ for all $j \in \set{j_1,\dots,j_n}$,\\
    $a$ fresh var.
\end{minipage}%
\end{center}

\begin{theorem}
  The strengthening rule is sound.
\end{theorem}
\begin{proof}
  In a resolution proof at level $k$, one can derive the proof objects $(\pobj \cup \set{j})^k$ for $j \in \set{j_1,\dots,j_n}$ using the conclusion of the strengthening rule.
  Assume we have a proof for $(\pobj \cup \set{i})^k$ (premise), then the quantified formula $\forall X_{k+1} \dots \quant X_n \ldot \bigwedge_{i^* \in \pobj} \literals(i^*,>k) \land \literals(i,>k)$ is false.
  Thus, the QBF with the same quantifier prefix and matrix, extended by some clause $\literals(j,>k)$ for $j \in \set{j_1,\dots,j_n}$, is still false.
  Since every $C_j$ subsumes $C_i$ with respect to quantifier level greater than $k$ ($\literals(j,>k) \subseteq \literals(i,>k)$), the clause $\literals(i,>k)$ can be eliminated without changing satisfiability.
  Thus, the resulting quantified formula $\forall X_{k+1} \dots \quant X_n \ldot \bigwedge_{i^* \in \pobj} \literals(i^*,>k) \land \literals(j,>k)$ is false and there exists a $\redres$ proof for $(\pobj \cup \set{j})^k$.
  \qed
\end{proof}

\begin{theorem} \label{thm:separation-strengthen-rule}
  The proof system without strengthening rule does not p-simulate the proof system with strengthening rule.
\end{theorem}
\begin{proof}
  We use the family of formulas $\crn$ that was used to show that level-ordered $\qres$ cannot $p$-simulate $\expres$~\cite{journals/tcs/JanotaM15}.
  We show that $\crn$ has a polynomial refutation in the $\redres$ calculus with strengthening rule, but has only exponential refutations without it.
  The latter follows from Theorem~\ref{thm:redres-level-ordered-qres-equivalent} and the results by Janota and Marques-Silva~\cite{journals/tcs/JanotaM15}.
  
  The formula $\crn$ has the quantifier prefix $\exists x_{11} \dots x_{nn} \forall z \exists a_1 \dots a_n b_1 \dots b_n$ and the matrix is given by
  \begin{equation} \label{eq:matrix_crn}
    \left( \bigvee_{i \in 1..n} \overline{a}_i \right) \land
    \left( \bigvee_{i \in 1..n} \overline{b}_i \right) \land
    \bigwedge_{i,j \in 1..n} \underbrace{(x_{ij} \lor z \lor a_i)}_{C_{ij}} \land \underbrace{(\overline{x}_{ij} \lor \overline{z} \lor b_j)}_{C_{\overline{ij}}}
  \end{equation}
  
  One can interpret the constraints as selecting rows and columns in a matrix where $i$ selects the row and $j$ selects the column, e.g., for $n=3$ it can be visualized as follows:
  \begin{equation*}
    \begin{array}{|l|l|l|l|l|l|}\hline
      x_{11} \lor z \lor a_1 &
      \overline{x}_{11} \lor \overline{z} \lor b_1 &
      x_{12} \lor z \lor a_1 &
      \overline{x}_{12} \lor \overline{z} \lor b_2 &
      x_{13} \lor z \lor a_1 &
      \overline{x}_{13} \lor \overline{z} \lor b_3 \\ \hline
      x_{21} \lor z \lor a_2 &
      \overline{x}_{21} \lor \overline{z} \lor b_1 &
      x_{22} \lor z \lor a_2 &
      \overline{x}_{22} \lor \overline{z} \lor b_2 &
      x_{23} \lor z \lor a_2 &
      \overline{x}_{23} \lor \overline{z} \lor b_3 \\ \hline
      x_{31} \lor z \lor a_3 &
      \overline{x}_{31} \lor \overline{z} \lor b_1 &
      x_{32} \lor z \lor a_3 &
      \overline{x}_{32} \lor \overline{z} \lor b_2 &
      x_{33} \lor z \lor a_3 &
      \overline{x}_{33} \lor \overline{z} \lor b_3 \\ \hline
    \end{array}
  \end{equation*}
  Assume $z \rightarrow 0$, then we derive the proof object $\pobj^1 = \set{i1 \mid i \in 1..n}^1$ ($\literals(\pobj^1) = \bigvee_{i \in 1..n} x_{i1}$) by applying the resolution and reduction rule.
  Likewise, for $z \rightarrow 1$, we derive the proof object $\pobj_0^1 = \set{ \overline{1j} \mid j \in 1..n }^1$ ($\literals(\pobj_0^1) = \bigvee_{j \in 1..n} \overline{x}_{1j}$).
  Applying the strengthening rule on $\pobj_0^1$ results in $\pobj_1^1 = (\set{c_1} \cup \set{ \overline{1j} \mid j \in 2..n})^1$ and $\set{\overline{c}_1, \overline{11}}^1, \set{\overline{c}_1, \overline{21}}^1, \dots, \set{\overline{c}_1, \overline{n1}}^1$ where $c_1$ is a fresh variable.
  Further $n-1$ applications of the strengthening rule starting on $\pobj_1^1$ lead to $\pobj_n^1 = \set{ c_j \mid j \in 1..n }^1$ and the proof objects $\set{\overline{c}_j, \overline{ij} \mid i,j \in 1..n}^1$,
  where $c_j$ are fresh variables, as all clauses in a column are equivalent with respect to the inner quantifiers (contain~$\overline{z} \lor b_j$).
  
  Using $\pobj^1$ and $\set{\overline{c}_1, \overline{11}}^1, \set{\overline{c}_1, \overline{21}}^1, \dots, \set{\overline{c}_1, \overline{n1}}^1$ from the first strengthening application, we derive the singleton set $\set{ \overline{c}_1 }$ using $n$ resolution steps ($\literals(\pobj^1) = \bigvee_{i \in 1..n} x_{i1}$ and $\literals(\set{\overline{c}_1, \overline{i1}}^1) = \set{\overline{c}_1, \overline{x}_{i1}}$).
  Analogously, one derives the singletons $\set{\overline{c}_2} \dots \set{\overline{c}_n}$ and together with $\pobj_n^1 = \set{ c_j \mid j \in 1..n }$ the empty set is derived.
  Thus, there exists a polyonomial resolution proof leading to a proof object $\pobj^0$ and the size of the overall proof is polynomial, too.
  \qed 
\end{proof}

We note that despite being stronger than plain $\redres$, the extended calculus is still incomparable to $\expres$.
\begin{corollary} \label{thm:strengthen-not-simulate-expres}
  $\redres$ with strengthening rule does not $p$-simulate $\expres$.
\end{corollary}
\begin{proof}
  We use a modification of formula $\crn$ (\ref{eq:matrix_crn}), which we call $\crn'$ in the following.
  The single universal variable $z$ is replaced by a number of variables $z_{ij}$ for every pair $i,j \in 1..N$.
  It follows that the strengthening rule is never applicable and hence, the proof system is as strong as level-ordered $\qres$ which has an exponential refutation of $\crn$ while $\expres$ has a polynomial refutation since the expansion tree has still only two branches~\cite{journals/tcs/JanotaM15}.
  \qed
\end{proof}

When compared to $\qres$, the strengthening rule can be interpreted as a step towards breaking the level-ordered constraint inherent to $\redres$.
The calculus, however, is not as strong as $\qres$.
\begin{corollary} \label{thm:strengthen-not-simulate-qres}
  $\redres$ with strengthening rule does not $p$-sim.~$\qres$.
\end{corollary}
\begin{proof}
	The formula $\crn'$ from the previous proof has a polynomial (tree-like) $\qres$ proof.
	The proof for $\crn$ given by Mahajan and Shukla~\cite{journals/ipl/MahajanS16} can be modified for $\crn'$.
	\qed
\end{proof}
Both results follow from the fact that the strengthening rule as presented is not applicable to the formula $\crn'$.
Where in $\crn$, the clauses $C_{\overline{ij}}$ are equal with respect to the inner quantifier when $j$ is fixed ($\overline{z} \lor b_j$), in $\crn'$ they are all different ($\overline{z}_{ij} \lor b_j$).
This difference is only due to the universal variables $z_{ij}$.
Thus, we propose a stronger version of the strengthening rule that does the subset check only on the existential variables.
For the universal literals, one additionally has to make sure that no resolvent produces a tautology (as it is the case in $\crn'$).
We leave the formalization to future work.

\subsection{Expansion} \label{sec:expansion-rule}

The levelized nature of the proof system allows us to introduce additional rules that can reason about quantified subformulas.
In the following, we introduce such a rule that allows us to use the $\expres$ calculus~\cite{journals/tcs/JanotaM15} within a $\redres$ proof.

We start by giving necessary notations used to define $\expres$.
We refer the reader to~\cite{journals/tcs/JanotaM15} for further information.
\begin{definition}[adapted from \cite{journals/tcs/JanotaM15}]
\begin{itemize}
  \item A \emph{$\forall$-expansion tree} for QBF $\Phi$ with $u$ universal quantifier blocks is a rooted tree $\tree$ such that every path $p_0 \xrightarrow{\alpha_1} p_1 \cdots \xrightarrow{\alpha_u} p_u$ in $\tree$ from the root $p_0$ to some leaf $p_u$ has exactly $u$ edges and each edge $p_{i-1} \xrightarrow{\alpha_i} p_i$ is labeled with a total assignment $\alpha_u$ to the universal variables at universal level $u$. Each path in $\tree$ is uniquely defined by its labeling.
  \item Let $\tree$ be a $\forall$-expansion tree and $P = p_0 \xrightarrow{\alpha_1} p_1 \cdots \xrightarrow{\alpha_u} p_u$ be a path from the root $p_0$ to some leaf $p_u$.
    \begin{enumerate}
      \item For an existential variable $x$ we define $\expandvar(P,x) = x^\assignment$ where $x^\assignment$ is a fresh variable and $\assignment$ is the universal assignment of the dependencies of $x$.
      \item For a propositional formula $\varphi$ define $\expand(P,\varphi)$ as instantiating $\varphi$ with $\alpha_1,\dots,\alpha_u$ and replacing every existential variable $x$ by $\expandvar(P,x)$.
      \item Define $\expand(\tree,\Phi)$ as the conjunction of all $\expand(P,\varphi)$ for each root-to-leaf $P$ in $\tree$.
    \end{enumerate}
\end{itemize}
\end{definition}
In contrast to previous work, we allow to use the expansion rule on quantified subformulas of $\Phi$ additionally to applying it to $\Phi$ directly.
By $\mathcal{C}^{\geq k}$ we denote a set of clauses that only contain literals bound at level $\geq k$.
\begin{align*}
  \infer[\expresrule]
  { \pobj^{k-1} }
  {\tree & \mathcal{C}^{\geq k} & \pi} \qquad
  &
  \begin{minipage}[c]{0.8\textwidth}
  $Q_k = \exists, \pi \text{ is a resolution refutation of the expansion}$\\
  formula $\expand(\tree, \exists X_k \ldot \forall X_{k+1} \dots \exists X_m \ldot \mathcal{C}^{\geq k})$\\
  $\pobj^{k-1} = \set{i \mid C_i \in \mathcal{C}}^{k-1}$
  \end{minipage}
\end{align*}
The rule states that if there is a universal expansion of the quantified Boolean formula $\exists X_k \ldot \forall X_{k+1} \dots \exists X_m \ldot \mathcal{C}^{\geq k}$ and a resolution refutation $\pi$ for this expansion, then there is no existential assignment that satisfies clauses $\mathcal{C}$ from level~$k$.
The size of the expansion rule is the sum of the size of the expansion tree and resolution proof~\cite{journals/tcs/JanotaM15}.

\begin{example} \label{ex:expres-rule}
  We demonstrate the interplay between $(\expresrule)$ and the $\redres$ calculus on the following formula
  \begin{align*}
    \overbrace{ \exists e_1 }^1 \ldot 
    \overbrace{ \forall u_1 }^2 \ldot
    \overbrace{ \exists c_1, c_2 }^3 \ldot
    \overbrace{ \forall a }^4 \ldot
    \overbrace{ \exists b \ldot
    \exists x }^5 \ldot
    \overbrace{ \forall z }^6 \ldot
    \overbrace{ \exists t }^7 \ldot \\ \hspace{-20pt}
    \underbrace{ (\overline{e}_1 \lor c_1) }_1
    \underbrace{ (\overline{u}_1 \lor c_1) }_2
    \underbrace{ (e_1 \lor c_2) }_3
    \underbrace{ (u_1 \lor c_2) }_4
    \underbrace{ (\overline{c}_1 \lor \overline{c}_2 \lor \overline{b} \lor \overline{a}) }_5
    \underbrace{ (z \lor t \lor b)  }_6
    \underbrace{ (\overline{z} \lor \overline{t})  }_7
    \underbrace{ (x \lor \overline{t})  }_8
    \underbrace{ (\overline{x} \lor t)  }_9
  \end{align*}
  
  To apply $(\expresrule)$, we use the clauses 5--9 from quantifier level 5, i.e., $\mathcal{C}^{\geq 5} = \set{
    (\overline{b})
    (z \lor t \lor b)
    (\overline{z} \lor \overline{t})
    (x \lor \overline{t})
    (\overline{x} \lor t)
  }$.
  The corresponding quantifier prefix is $\exists b \exists x \forall z \exists t$.
  Using the complete expansion of $z$ ($\set{z \to 0, z \to 1}$) as the expansion tree $\tree$, we get the following expansion formula
  \begin{equation*}
    (\overline{b})
    (t^\set{z \to 0} \lor b)
    (x \lor \overline{t}^\set{z \to 0})
    (\overline{x} \lor t^\set{z \to 0})
    (\overline{t}^\set{z \to 1})
    (x \lor \overline{t}^\set{z \to 1})
    (\overline{x} \lor t^\set{z \to 1}) \enspace,
  \end{equation*}
  which has a simple resolution proof $\pi$.
  The conclusion of $(\expresrule)$ leads to the proof object $\set{ 5,6,7,8,9 }^4$, but only clause 5 contains literals bound before quantification level $5$.
  After a universal reduction, the proof continues as described in Example~\ref{ex:example_redres_proof}.
\end{example}

\begin{theorem}
  The $\forall$exp-res rule is sound.
\end{theorem}
\begin{proof}
  Assume otherwise, then one would be able to derive a proof object $\pobj^{k-1}$ that is part of a $\redres$ refutation proof for true QBF $\Phi$.
  Thus, the clause corresponding to $\pobj^{k-1}$ (cf.~proof of Theorem~\ref{thm:redres-soundness}) $(\bigvee_{i \in \pobj} \bigvee_{l \in \literals(i, < k)} l)$ made $\Phi$ false.
  However, the same clause can be derived directly by applying the expansion $\tree$ to the original QBF, i.e., expanding universal variables beginning with quantification level $k+1$, and propositional resolution on the resulting expansion formula.
  Thus, this clause can be conjunctively added to the matrix without changing satisfiability, leading to a contradiction.
  \qed
\end{proof}

The resulting proof system can be viewed as a unification of the currently known CEGAR approaches for solving quantified Boolean formulas~\cite{journals/ai/JanotaKMC16,conf/ijcai/JanotaM15,conf/fmcad/RabeT15}.
\begin{theorem} \label{thm:expres-vs-redexpres}
  $\expres$ does not $p$-simulate $\redexpres$.
\end{theorem}
\begin{proof}
  $\expres$ does not $p$-simulate level-ordered $Q$-resolution~\cite{journals/ipl/MahajanS16}.
  \qed
\end{proof}
The combination of both rules makes the proof system stronger than merely choosing between expansion and resolution proof upfront.
\begin{theorem} \label{thm:combined-stronger-than-applying-independently}
  There is a family of quantified Boolean formulas that have polynomial refutation in $\redexpres$, but have only exponential refutations in $\redres$ and $\expres$.
\end{theorem}
\begin{proof}
  For this proof, we take two formulas that are hard for $\qres$ and $\expres$, respectively.
  We build a new family of formulas that has a polynomial refutation in $\redexpres$, but only exponential refutations in $\redres$ and $\expres$.
  
  The first formula we consider is formula (2) form~\cite{journals/tcs/JanotaM15}, that we call $\dagn$ in the following:
  \begin{align*}
    &\exists e_1 \forall u_1 \exists c_1 c_2 \dots
    \exists e_n \forall u_n \exists c_{2n-1} c_{2n} \ldot \\
    & (\bigvee_{i \in 1 \dots 2n} \overline{c}_i) \land \bigwedge_{i \in 1 \dots n} (\overline{e}_i \lor c_{2i-1}) \land (\overline{u}_i \lor c_{2i-1}) \land (e_i \lor c_{2i}) \land (u_i \lor c_{2i})
  \end{align*}
  It is known that $\dagn$ has a polynomial level-ordered $\qres$ proof and only exponential $\expres$ proofs~\cite{journals/tcs/JanotaM15}.
  As a second formula, we use the $\qparity_n$ formula~\cite{conf/stacs/BeyersdorffCJ15}
  \begin{equation*}
    \exists x_1 \dots x_n \forall z \exists t_2 \dots t_n \ldot
    \xor(x_1,x_2,t_2) \land
    \bigwedge_{i \in 3 \dots n} \xor(t_{i-1}, x_i, t_i) \land (z \lor t_n) \land (\overline{z} \lor \overline{t}_n)
  \end{equation*}
  where $\xor(o_1, o_2, o) = (\overline{o}_1 \lor \overline{o}_2 \lor \overline{o}) \land (o_1 \lor o_2 \lor \overline{o}) \land (\overline{o}_1 \lor o_2 \lor o) \land (o_1 \lor \overline{o}_2 \lor o)$ defines $o$ to be equal to $o_1 \oplus o_2$.
  $\qparity_n$ has a polynomial $\expres$ refutation but only exponential $\qres$ refutations~\cite{conf/stacs/BeyersdorffCJ15}.
  We construct the following formula
  \begin{align*}
    &
    \exists e_1 \forall u_1 \exists c_1 c_2 \dots
    \exists e_n \forall u_n \exists c_{2n-1} c_{2n} \ldot
    \forall a \exists b \ldot 
    \exists x_1 \dots x_n \forall z \exists t_2 \dots t_n \ldot\\
    &
    \bigwedge_{i \in 1 \dots n} (\overline{e}_i \lor c_{2i-1})
    {}\land (\overline{u}_i \lor c_{2i-1})
    {}\land (e_i \lor c_{2i})
    {}\land (u_i \lor c_{2i}) \land{} \\
    &
    (\overline{a} \lor \overline{b} \lor \hspace{-6pt} \bigvee_{i \in 1 \dots 2n} \hspace{-6pt} \overline{c}_i) \land
    \xor(x_1,x_2,t_2)
    {}\land \hspace{-6pt}\bigwedge_{i \in 3 \dots n} \hspace{-6pt} \xor(t_{i-1}, x_i, t_i)
    {}\land (z \lor t_n \lor b)
    {}\land (\overline{z} \lor \overline{t}_n)
  \end{align*}
  We argue in the following that this formula has a polynomial refutation in $\redexpres$.
  First, using $(\expresrule)$ we can derive the proof object containing the clause $(\overline{a} \lor \bigvee_{i \in \set{1 \dots 2n}} \overline{c}_i)$ using the expansion tree $\tree = \set{z \rightarrow 0, z \rightarrow 1}$ and the clauses from the last row (analogue to Example~\ref{ex:expres-rule}).
  After applying universal reduction, the proof object representing clause $(\bigvee_{i \in \set{1 \dots 2n}} \overline{c}_i)$ can be derived.
  For the remaining formula, there is a polynomial and level-ordered resolution proof~\cite{journals/tcs/JanotaM15}, thus, the formula has a polynomial $\redexpres$ proof.
  
  There is no polynomial $\qres$ proof, because deriving $(\bigvee_{i \in \set{1 \dots 2n}} \overline{c}_i)$ is exponential in $\qres$.
  Likewise, there is no polynomial $\expres$ proof as the formula after deriving this clause has only exponential $\expres$ refutations.
  \qed
\end{proof}

One question that remains open, is how the new proof system compares to unrestricted $\qres$.
We already know that the new proof system polynomially simulates both tree-like $\qres$ as well as level-ordered $\qres$.

\begin{theorem} \label{thm:redexpres-not-simulate-qres}
  $\redexpres$ does not $p$-simulate $\qres$.
\end{theorem}
\begin{proof}[Sketch]
  We construct a formula that is hard for expansion and level-ordered $\qres$, but easy for (unrestricted) $\qres$.
  We have already seen in the proof of Theorem~\ref{thm:combined-stronger-than-applying-independently} that $\dagn$ is hard for $\expres$ but easy for $\qres$.
  However, the $\qres$ proof of $\dagn$ is level-ordered.
  Hence, we need an additional formula that is hard to refute for level-ordered $\qres$.
  We use the modified pigeon hole formula from~\cite{journals/amai/Goerdt92} where unrestricted resolution has polynomial proofs and resolution proofs that are restricted to a certain variable ordering are exponential.
  Using universal quantification, one can impose an arbitrary order on a level-ordered $\qres$ proof, thus, there is a quantified Boolean formula which has only exponential level-ordered $\qres$ but has a polynomial $\qres$ proof.
  The disjunction of those two formulas gives the required witness.
  This formula is easy to refute for $\qres$, but the first one is hard for $\expres$ and the second is hard for level-ordered $\qres$.
  \qed
\end{proof}

\subsection{Comparison Between Extensions}
We conclude this section by comparing the two extensions of the $\redres$ calculus introduced in this paper.

\begin{theorem} \label{thm:extensions-incomparable}
	$\redexpres$ and $\redres$ with strengthening rule are incomparable.
\end{theorem}
\begin{proof}[Sketch]
The family of formulas $\crn'$ from proof of Corollary~\ref{thm:strengthen-not-simulate-expres} separates $\redexpres$ and $\redres$ with strengthening rule.
Since the strengthening rule is not applicable, all $\redres$ proofs are exponential while there is a polynomial proof in $\redexpres$.

The other direction is shown by using a similar construction as the one used in the proof of Theorem~\ref{thm:combined-stronger-than-applying-independently}.
We use a combination of $\crn$ and $\dagn$ to construct a formula that has only exponential refutations in $\redexpres$, but a polynomial refutation using the strengthening rule.
The formula $\dagn$ is used to generate the premise for the application of the strengthening rule to solve $\crn$.
To generate this premise using the rule $(\expresrule)$ one needs an exponential proof.
There is a polynomial proof for $\dagn$ in $\redres$, but there is none for $\crn$, thus, $\redexpres$ has only exponential refutations.
  \qed
\end{proof}

\begin{theorem} \label{thm:redexpres-strengthen-qres-incomparable}
	$\redexpres$ with strengthening rule and $\qres$ are incomparable.
\end{theorem}
\begin{proof}
	Follows from the proof of Theorem~\ref{thm:redexpres-not-simulate-qres} as the witnessing formula can be constructed such that the strengthening rule is not applicable.
	The other direction follows from the separation of $\qres$ and $\expres$ by Beyersdorff et al.~\cite{conf/stacs/BeyersdorffCJ15}.
	\qed
\end{proof}

\section{Experimental Evaluation}

\subsection{Implementation}

We extended the implementation of $\caqe$ with the possibility to use the rule $(\expresrule)$ as introduced in Sec.~\ref{sec:expansion-rule}\footnote{$\caqe$ is available online at \url{https://react.uni-saarland.de/tools/caqe/}.}.
While the rule is applicable at every level in the QBF in principle, the effectiveness decreases when applying it to deeply nested formulas where $\caqe$ tends to perform better~\cite{conf/fmcad/RabeT15} than $\rareqs$.
We aim to strike a balance between expansion and clausal-abstraction, i.e., keeping the best performance characteristics of both solving methods.
Thus, in our implementation, we apply the expansion refinement (additional to the clausal-abstraction refinement) to the innermost universal quantifier.

\begin{algorithm}[t]
\begin{algorithmic}[1]
\State $\varphi_k$ is the propositional abstraction for quantifier $\exists X_k$
\Procedure{solve$_\exists$}{$\exists X_k \ldot\Psi$, $\pobj^k$}
\While{$\mathit{true}$}
  \State disable clauses $C_i^k$ of $\varphi_k$ where $i \notin \pobj^k$ \Comment{those $C_i$ are already satisfied $<k$}
  \State generate candidate solution $\pobj_*^{k+1}$ using SAT solver and abstraction $\varphi_k$ \label{alg:cegar-algorithm_generate_candidate}
  \If{no candidate exists} \Comment{there is a resolution proof $\pi$}
    \State \Return UNSAT, $\pobj^{k-1}$ \label{alg:unsat-result}
  \ElsIf{$\Psi$ is propositional} \Comment{base case for structural recursion}
    \State \Return SAT, witness
  \EndIf
  \State verify candidate recursively, call \Call{solve$_\forall$}{$\Psi$, $\pobj_*^{k+1}$} \label{alg:cegar-algorithm_verify_recursively}
  \If{candidate correct}
    \State \Return SAT, witness
  \Else
    \State counterexample consists of $\pobj_\text{ce}^k$ and \textbf{expansion tree $\tree$}
    \State refine $\varphi_k$ such that one clause $C_i^k$ in with $i \in \pobj_\text{ce}$ must be satisfied
    \State \textbf{refine $\varphi_k$ with abstraction of expansion of $\Phi$ with respect to $\tree$}
  \EndIf
\EndWhile
\EndProcedure
\end{algorithmic}
\caption{Modified CEGAR solving loop for existential quantifier}
\label{alg:cegar-algorithm}
\end{algorithm}

An overview of the CEGAR algorithm is given in Algorithm~\ref{alg:cegar-algorithm}.
There is a close connection between the rules of the $\redres$ calculus and the presented algorithm.
Especially, we use a SAT solver to prove the refutation needed in the rule $(\resrule)$.
We refer to~\cite{conf/fmcad/RabeT15} for algorithmic details.
Changes to the original algorithm are written in bold text.

\paragraph{Abstraction.}
The abstraction for quantifier $\exists X_k$, written $\varphi_k$ is the projection of the clauses of the matrix to variables in $X_k$, i.e., $\bigwedge_{1 \leq i \leq m} \literals(i,k)$.
We assume that there is a operation to ``disable'' clauses in $\varphi_k$ which corresponds to the situation where a clause $C_i$ is satisfied by some variable bound before~$k$.
Likewise, for every clause we allow the assumption that this clause will be satisfied by a some variable bound after $k$.
This is used to generate candidate proof objects $\pobj_*^{k+1}$ for inner levels.
In the refinement step, this assumption can be invalidated, i.e., there is a way to force satisfaction of a clause at level $k$.
Those operations can be implemented by an incremental SAT solver and two additional literals controlling the satisfaction of clauses~\cite{conf/fmcad/RabeT15}.

\paragraph{Algorithm.}
The algorithm recurses on the structure of the quantifier prefix and communicates proof objects $\pobj$, which indicate the clauses of the matrix that are satisfied.
At an existential quantifier, the abstraction generates a candidate solution~(line~\ref{alg:cegar-algorithm_generate_candidate}) and checks recursively whether the candidate is correct~(line~\ref{alg:cegar-algorithm_verify_recursively}).
If not, the counterexample originally consists of a set of clauses (which could not be satisfied from the inner existential quantifiers).
We extend this counterexample to also include an expansion tree $\tree$ from the levels below.
Additionally to the original refinement, we also build the expansion of the QBF with respect to the expansion tree $\tree$, resulting in a QBF with the same quantifier prefix as the current level (with additional existential variables due to expansion).
This QBF is then translated into a propositional formula in the same way as the original QBF.
Lastly, the abstraction $\varphi_k$ is then conjunctively combined with this propositional formula.
Note that if the function returns UNSAT (line~\ref{alg:unsat-result}), the corresponding resolution proof from the SAT solver can be used to apply the rule $(\resrule)$ form the $\redres$ calculus.

As the underlying SAT solver in the implementation, we use $\picosat$~\cite{journals/jsat/Biere08}, $\minisat$~\cite{conf/sat/EenS03}, $\cmsat$~\cite{conf/sat/SoosNC09}, or Lingeling~\cite{conf/sat/Biere14}.

\subsection{Evaluation}

In our evaluation, we show that the established theoretical separations shown in the last section translate to a significant empirical improvement.
The evaluation is structured by the following three hypothesizes:
First, the strengthen and expansion refinement give a significant improvement over the plain version of $\caqe$.
Combining both refinements is overall better than only applying one of them.
Second, we show that the improvement provided by the those refinements is independently of the underlying SAT solver.
Third, when comparing on a per instance basis, the combined refinement effects the runtime mostly positively.
We show that the improvement is up to three orders of magnitude.

\begin{table}[t]
  \caption{Number of solved instances of the QBFGallery~2014 and QBFEval~2016 benchmark sets.}
  \label{tbl:qbfgallery2014}
  \centering
  \begin{tabular}{lrrrrrrrrr}
    \hline\noalign{\smallskip}
    
    family & total & \multicolumn{4}{c}{$\caqe$-$\cmsat$} & $\rareqs$ & $\qesto$ & $\depqbf$ & $\ghostq$ \\
    && {\tiny plain} & {\tiny strengthen} & {\tiny expansion} & {\tiny both} \\
    \noalign{\smallskip}
    \hline
    \noalign{\smallskip}
    eval2012r2 & 276 & 128 & 129 & 146 & \textbf{149} & 134 & 132 & 139 & 145 \\
    bomb       & 132 & 94 & \textbf{95} & 94 & 94 &  82 & 78 & 80 & 82 \\
    complexity & 104 & 60 & 68 & 86 & 85 & \textbf{90} & 76 & 51 & 43 \\
    dungeon    & 107 & 60 & 65 & \textbf{70} & \textbf{70} & 61 & 57 & 67 & 50 \\
    hardness   & 114 & 108 & 102 & \textbf{109} & 101 & 69 & 106 & 80 & 51 \\
    planning   & 147 & 45 & 93 & 65 & 95 & \textbf{144} & 55 & 38 & 13 \\
    testing    & 131 & 91 & 86 & 93 & 91 & 95 & 90 & 99 & \textbf{113} \\
    preprocessing & 242 & 86 & 93 & 105 & \textbf{110} & 107 & 104 & 108 & 60 \\
    \noalign{\smallskip}
    \hline
    \noalign{\smallskip}
    gallery2014 & 1253 & 672 & 731 & 768 & \textbf{795} & 782 & 698 & 662 & 557 \\
    \noalign{\smallskip}
    \hline
    \noalign{\smallskip}
    eval2016   & 825 & 607 & 611 & 635 & 636 & \textbf{644} & 623 & 598 & 595 \\
    \noalign{\smallskip}
    \hline
    \hline
    \noalign{\smallskip}
    all        & 2078 & 1279 & 1342 & 1403 & \textbf{1431} & 1426 & 1321 & 1260 & 1152 \\
    \noalign{\smallskip}
    \hline
  \end{tabular}
\end{table}

We compare our implementation against $\rareqs$~\cite{journals/ai/JanotaKMC16}, $\qesto$~\cite{conf/ijcai/JanotaM15}, $\depqbf$ in version 5.0~\cite{journals/jsat/LonsingB10}, and $\ghostq$~\cite{conf/sat/KlieberSGC10}.
For every solver except $\ghostq$, we use $\bloqqer$~\cite{conf/cade/BiereLS11} in version 031 as preprocessor.
For our experiments, we used a machine with a $3.6\,\text{GHz}$ quad-core Intel Xeon processor and $32\,\text{GB}$ of memory.
The timeout and memout were set to $10$ minutes and $8\,\text{GB}$, respectively.
Table~\ref{tbl:qbfgallery2014} shows number of solved instances on the QBFGallery~2014 benchmark set, broken down by benchmark family, as well as the more recent QBFEval~2016 benchmark set.
For $\caqe$, we only report on the best performing version, that is the one using $\cmsat$ as a backend solver.

The table shows that the strengthen and expansion refinement individually improve over the plain version of $\caqe$ in the number of solved instances.
Further, the combination of both refinements is the overall best solver, followed by $\rareqs$.

In the following, we refer to the combination of strengthen and expansion refinement as extended refinements.
We want to detail the improvements due to the extended refinements and show their independence of the backend solver.
The plot in Fig.~\ref{fig:internal_comparison} depicts the effect of the extended refinements with respect to the solved instances.
The improvements in the number of solved instances are independent from the choice of the underlying SAT solver and range between $100$ to $150$ more instances solved compared to the plain version of $\caqe$.

\begin{figure}[t]
  \centering
  \begin{tikzpicture}
    \begin{axis}[xlabel=\# solved instances,ylabel=time (sec.),width=.95\columnwidth,height=7.5cm,ymin=0,ymax=600,xmin=200,xmax=810,legend entries={$\cmsat$,+expansion,$\lingeling$,+expansion,$\picosat$,+expansion,$\minisat$,+expansion},mark size=1.8pt,
        legend style={
          at={(.33,0.95)},
          anchor=north east,
          draw=none}]
        ]]
        
      \addplot+[blue,solid,mark=*,mark options={fill=blue}] table {plots/caqe-expansion-qbfgallery-2014_cactus_caqe-cmsat-g0.dat};
      \addplot+[blue,dashed,mark=o,mark options={fill=blue,solid}] table {plots/caqe-expansion-qbfgallery-2014_cactus_caqe-cmsat-expansion-strong-unsat-g0.dat};
      
      \addplot+[red,solid,mark=square*,mark options={fill=red}] table {plots/caqe-expansion-qbfgallery-2014_cactus_caqe-lingeling-g0.dat};
      \addplot+[red,dashed,mark=square,mark options={fill=red,solid}] table {plots/caqe-expansion-qbfgallery-2014_cactus_caqe-lingeling-expansion-strong-unsat-g0.dat};
      
      \addplot+[green,solid,mark=triangle*,mark options={fill=green}] table {plots/caqe-expansion-qbfgallery-2014_cactus_caqe-picosat-g0.dat};
      \addplot+[green,dashed,mark=triangle,mark options={fill=green,solid}] table {plots/caqe-expansion-qbfgallery-2014_cactus_caqe-picosat-expansion-strong-unsat-g0.dat};
      
      \addplot+[orange,solid,mark=diamond*,mark options={fill=orange}] table {plots/caqe-expansion-qbfgallery-2014_cactus_caqe-minisat-g0.dat};
      \addplot+[orange,dashed,mark=diamond,mark options={fill=orange,solid}] table {plots/caqe-expansion-qbfgallery-2014_cactus_caqe-minisat-expansion-strong-unsat-g0.dat};
    \end{axis}
  \end{tikzpicture}
  \caption{Effect of the expansion refinement on the different configurations of $\caqe$ on the GBFGallery~2014 benchmark sets.}
  \label{fig:internal_comparison}
\end{figure}
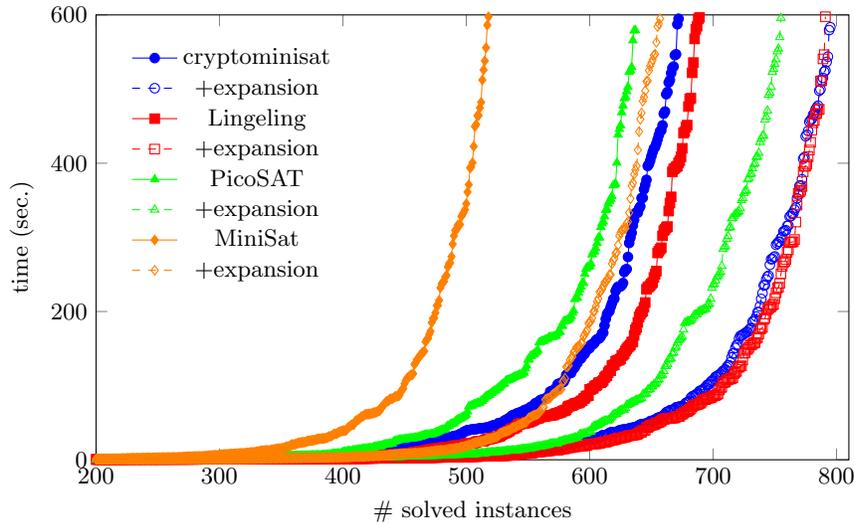

The scatter plot depicted in Fig.~\ref{fig:scatter_expansion} compares the running times of plain $\caqe$ to the one using extended refinements (both using $\cmsat$) on a per instance basis.
Marks below the diagonal means that the variant using extended refinements is faster.
It is remarkable that the extended refinements have mostly positive effect on the solving times.
Only a few instances saw a significant increase in solving time and even less timed out with extended refinements while being solved before.
On the other hand, we see improvements in solving time that exceed three orders of magnitude.
This is an empirical confirmation of our goal stated before that our implementation of expansion-refinement adds performance characteristic of expansion-based solvers while keeping the characteristics of the clausal-abstraction algorithm.

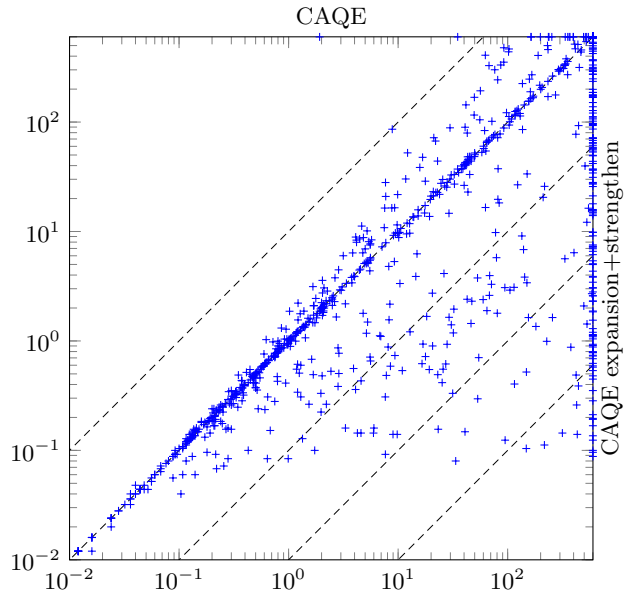
\begin{figure}[t]
  \centering
  \begin{tikzpicture}
    \begin{loglogaxis}[
      small,
      only marks,
      axis equal,
      xmin=0.01,xmax=601,
      ymin=0.01,ymax=601,
      enlargelimits=false,
      width=.7\columnwidth,
      height=.7\columnwidth,
      mark size=1.5pt,
      xlabel=$\caqe$,
      xlabel style={
        at={(axis description cs:0.5,1)},
        anchor=south,
      },
      ylabel=$\caqe$ expansion+strengthen,
      ylabel style={
        at={(axis description cs:1,0.5)},
        anchor=north,
      }]
      \addplot+[mark=+] table {plots/caqe-expansion-qbfgallery-2014_scatter_caqe-cmsat_vs_caqe-cmsat-expansion-strong-unsat.dat};
      
      \draw[densely dashed,very thin] (0.001,0.001) -- (601,601);
      \draw[densely dashed,very thin] (0.0001,0.001) -- (60.1,601);
      \draw[densely dashed,very thin] (0.001,0.0001) -- (601,60.1);
      \draw[densely dashed,very thin] (0.001,0.00001) -- (601,6.01);
      \draw[densely dashed,very thin] (0.001,0.000001) -- (601,0.601);
    \end{loglogaxis}
  \end{tikzpicture}
  \caption{Scatter plot comparing the solving time (in sec.) of $\caqe$ with and without extended refinement.}
  \label{fig:scatter_expansion}
\end{figure}

\section{Related Work}

$\qres$~\cite{journals/iandc/BuningKF95} is a variant of propositional refutation that is sound and refutation complete for QBF.
There have been extensions proposed to $\qres$, like long-distance resolution~\cite{conf/iccad/ZhangM02} and universal resolution~\cite{conf/cp/Gelder12}, some which are implemented in the QCDCL solver $\depqbf$~\cite{journals/jsat/LonsingB10}.
Recently, there has also been extensions proposed that extend $\qres$ by more generalized axioms~\cite{conf/sat/LonsingES16}.
In some sense, the $(\expresrule)$ rule presented in this paper can be viewed as an new axiom rule for the $\redres$ calculus.

The $\expres$ calculus~\cite{journals/tcs/JanotaM15} was introduced to allow reasoning over expansion-based QBF solving, exemplified by the QBF solver $\rareqs$~\cite{journals/ai/JanotaKMC16}.
The work on $\redres$ was motivated by the same desire, namely understanding the performance of the recently introduced QBF solvers $\caqe$~\cite{conf/fmcad/RabeT15} and $\qesto$~\cite{conf/ijcai/JanotaM15}.
The incomparability of $\expres$ and $\redres$~\cite{journals/tcs/JanotaM15,conf/stacs/BeyersdorffCJ15} lead to the creation of stronger proof systems that unify those calculi, like $\ircalc$~\cite{conf/mfcs/BeyersdorffCJ14}.
Further separation results, between variants of $\ircalc$ and variants of $\qres$, were given in~\cite{conf/stacs/BeyersdorffCJ15}.
Those extensions, however, do not have accompanying implementations.
This also applies to recent work that is based on first-order resolution~\cite{conf/sat/Egly16}.

There are two well-known restrictions to $\qres$, that is level-ordered and tree-like $\qres$.
Those restricted calculi were shown to be incomparable~\cite{journals/ipl/MahajanS16}.
QCDCL based solver exhibit level-ordered proofs~\cite{conf/sat/Janota16} and it was shown that $\expres$ $p$-simulates tree-like $\qres$~\cite{journals/tcs/JanotaM15}.
We showed that $\redres$ is polynomial simulation equivalent to level-ordered $\qres$, which explains similar performance characteristics of the underlying solvers.
Further, the strengthening rule presented in this paper can be viewed as a first step towards breaking the level-ordered restriction.
The $\redexpres$ calculus $p$-simulates level-ordered and tree-like $\qres$.

\section{Conclusion}

In this paper, we have introduced a new QBF proof calculus $\redres$ and showed that it is suitable for describing CEGAR based solving algorithms.
We defined two extensions of the $\redres$ calculus and showed that there is a theoretical advantage over the basic calculus.
Based on this foundation, we implemented an expansion refinement in the solver $\caqe$ and evaluated it on standard QBF benchmark sets.
Our experiments show that our new implementation significantly outperforms the previous one, with little to no negative impact, making it one of the most competitive QBF solver available.
We have also shown that our theoretical considerations and the consequent algorithmic change explains those practical gains.

In future work, we want to improve the implementation by exploring heuristics for the application of the different refinements and we want to explore alternative versions of the strengthening rule presented in this paper.

\paragraph{Acknowledgments.}
I thank Christopher Hahn and the anonymous reviewers for their comments on earlier versions of this paper.

\bibliographystyle{splncs03}
\bibliography{main}

\end{document}